\documentclass{IEEEtran}
\usepackage{graphicx}
\usepackage{balance}  
\usepackage{times,amssymb,amsmath,amsfonts,epsfig,latexsym,stmaryrd}
\usepackage{listings}
\usepackage{algorithm}
\usepackage{algorithmic}
\usepackage{subfig}
\usepackage{color}
\usepackage{paralist}

\newtheorem{definition}{Definition}
\newtheorem{theorem}{Theorem}
\newtheorem{proposition}{Proposition}

\newtheorem{example}{Example}

\newcommand{\se}{\!=\!}
\newcommand{\ignore}[1]{}

\begin{document}


\title{Making massive probabilistic databases practical}

\author{Andrei~Todor,
        Alin~Dobra,
        Tamer~Kahveci, Christopher Dudley
\IEEEcompsocitemizethanks{\IEEEcompsocthanksitem A. Todor, A. Dobra, T. Kahveci and Christopher Dudley are with the Department
of Computer and Information Sciences and Engineering, University of Florida, FL, 32611-6120.\protect\\
E-mail: \{atodor,adobra,tamer\}@cise.ufl.edu }
\thanks{}}

\maketitle

\begin{abstract}
  Existence of incomplete and imprecise data has moved the database
  paradigm from deterministic to probababilistic information.
  Probabilistic databases contain tuples that may or may not exist
  with some probability.  As a result, the number of possible
  deterministic database instances that can be observed from a
  probabilistic database grows exponentially with the number of
  probabilistic tuples. In this paper, we consider the problem of
  answering both aggregate and non-aggregate queries on massive
  probabilistic databases.  We adopt the tuple independence model, in
  which each tuple is assigned a probability value.  We develop a
  method that exploits Probability Generating Functions (PGF) to
  answer such queries efficiently.  Our method maintains a polynomial
  for each tuple. It incrementally builds a master polynomial that
  expresses the distribution of the possible result values precisely.
  We also develop an approximation method that finds the distribution
  of the result value with negligible errors. Our
  experiments suggest that our methods are orders of magnitude faster
  than the most recent systems that answer such queries, including
  MayBMS and SPROUT\@.  In our experiments, we were able to scale up to
  several terabytes of data on TPC-H queries, while existing methods
  could only run for a few gigabytes of data on the same queries.
\end{abstract}

\section{Introduction}
Today many applications have to deal with a flood of data acquired by means of new sensing technologies. These techologies gather data at enormous rates, but are not always accurate. Organizing, querying or mining such uncertain data is imperative in modern database systems. 

The reason for an urgent need for efficient querying mechanisms for uncertain databases is not limited to their ever-growing sizes. Traditional databases are also often confronted with tasks that require a probabilistic treatment of their records. One example scenario is cross-cutting queries involving both deterministic and probabilistic data. For example, consider a query which has to join a large deterministic table and a small probabilistic table. The resulting table is probabilistic, and this characteristic has to be propagated through the subsequent stages of the computation. Another area where probabilistic approaches benefit deterministic data are fuzzy queries~\cite{testemale,zemankova}. When confronted with a decision, it is seldom the case that we have a hard criterion. Therefore, for a selection predicate, it is desirable to allow a soft margin for the tuples to pass the selection. This can be achieved by returning not just the tuple itself, but also a measure of how close the tuple is to the margin. This measure can be expressed as a probability, given by a function similar to a sigmoid around the threshold value. Querying the resulting data requires a probabilistic approach. Another scenario where a probabilistic assessment is needed is apparent when ranking query results in deterministic databases~\cite{agrawal,chaudhuri}. Large databases often pose a challenge to the user by returning an overwhelming number of tuples. This happens particularly when the selection criteria are too indiscriminate. To rank the results in a large answer set we need a probabilistic measure of the relevance of each returned tuple. On the contrary, when the criteria are too restrictive, the query may return no results. In this case, we can use a fuzzy query approach to return the tuples that are closest to the desired result. All these current challenges to modern database systems prompt us to look for efficient solutions for queries in large probabilistic databases.

To distinguish traditional terminology from the probabilistic approach of this paper, we call an uncertain tuple a \emph{probabilistic tuple} or \emph{probabilistic data}. Each probabilistic tuple in a relation may or may not correspond to factual data in reality with some probability. We call the database systems that manage probabilistic data \emph{ probabilistic databases}.

In the simple tuple independent model~\cite{suciu}, which we adopt in this paper, the probability of the tuple is an attribute of that tuple and is independent of all other tuples. A table with $n$ probabilistic tuples leads to $2^n$ possible deterministic tables, as each tuple may or may not appear in a deterministic table. Figure~\ref{ex} presents a small table with only three probabilistic tuples and one attribute. The second column is the tuple probability. Figure~\ref{ex2} presents the eight deterministic tables that can be instances of the probabilistic table. The deterministic instances have a probability of occurence depending on which tuples are present, and are mutually exclusive. The enumeration of all possible deterministic instantiations is known as the possible worlds model~\cite{HavingRe}, and is of purely theoretical importance. In fact, all related work on querying probabilistic databases aims to avoid enumerating the possible determinstic databases resulting from a probabilistic database. Nevertheless, the exponential growth in the possible worlds semantic is a way of evidentiating why querying probabilistic databases is very difficult. 

\begin{figure}
\caption{A probabilistic table}
\centering
\begin{tabular}{c|c|c}
\hline
\multicolumn{1}{|c|}{R} & v & \multicolumn{1}{c|}{p} \\
\hline
  & 3 & \multicolumn{1}{c|}{0.7} \\
  & 8 & \multicolumn{1}{c|}{0.8} \\
  & 5 & \multicolumn{1}{c|}{0.5} \\ \cline{2-3}
\end{tabular}
\label{ex}
\end{figure}

\begin{figure}
\caption{Deterministic tables that can be instances of the probabilistic table in Figure~\ref{ex}}
\centering
\subfloat[p=0.03]{
\begin{tabular}{c|c}
\hline
\multicolumn{1}{|c|}{R} & \multicolumn{1}{c|}{v} \\
\hline
\end{tabular}}
\qquad 
\subfloat[p=0.07]{
\begin{tabular}{c|c}
\hline
\multicolumn{1}{|c|}{R} & \multicolumn{1}{c|}{v} \\
\hline
  & \multicolumn{1}{c|}{3} \\\cline{2-2}
\end{tabular}}
\qquad
\subfloat[p=0.12]{
\begin{tabular}{c|c|c}
\hline
\multicolumn{1}{|c|}{R} & \multicolumn{1}{c|}{v} \\
\hline
  & \multicolumn{1}{c|}{8} \\\cline{2-2}
\end{tabular}}
\qquad
\subfloat[p=0.03]{
\begin{tabular}{c|c}
\hline
\multicolumn{1}{|c|}{R} & \multicolumn{1}{c|}{v} \\
\hline
  & \multicolumn{1}{c|}{5} \\\cline{2-2}
\end{tabular}}
\qquad
\subfloat[p=0.28]{
\begin{tabular}{c|c}
\hline
\multicolumn{1}{|c|}{R} & \multicolumn{1}{c|}{v} \\
\hline
  & \multicolumn{1}{c|}{3} \\
  & \multicolumn{1}{c|}{8} \\\cline{2-2}
\end{tabular}}
\qquad 
\subfloat[p=0.07]{
\begin{tabular}{c|c}
\hline
\multicolumn{1}{|c|}{R} & \multicolumn{1}{c|}{v} \\
\hline
  & \multicolumn{1}{c|}{3} \\
  & \multicolumn{1}{c|}{5} \\ \cline{2-2}
\end{tabular}}
\qquad
\subfloat[p=0.12]{
\begin{tabular}{c|c}
\hline
\multicolumn{1}{|c|}{R} & \multicolumn{1}{c|}{v} \\
\hline
  & \multicolumn{1}{c|}{8} \\
  & \multicolumn{1}{c|}{5} \\ \cline{2-2}
\end{tabular}}
\qquad
\subfloat[p=0.28]{
\begin{tabular}{c|c}
\hline
\multicolumn{1}{|c|}{R} & \multicolumn{1}{c|}{v} \\
\hline
  & \multicolumn{1}{c|}{3} \\
  & \multicolumn{1}{c|}{8} \\
  & \multicolumn{1}{c|}{5} \\ \cline{2-2}
\end{tabular}}
\qquad
\label{ex2}
\end{figure}

The rest of the paper is organized as follows. In section ~\ref{sec:related} we briefly survey the existing work in the area. In Section~\ref{sec:contributions} we provide an overview of our contributions. In Section~\ref{sec:math} we introduce the theoretical foundations for our method. In Section~\ref{sec:queries} we describe our process of transforming probabilistic queries into deterministic query plans. In Section~\ref{sec:imp} we present the implementation details for our method. In Section~\ref{sec:exp} we report our experimental results, and the last section is dedicated to conclusions.

\section {Related Work}
\label{sec:related}
From the seminal work of Imielinski and Lipski~\cite{imielinski} on incomplete information stemmed a new area of research trying to provide solutions to the tasks related to probabilistic data. These solutions propose a series of models such as v-tables, c-tables~\cite{imielinski}, pc-tables~\cite{green2006}, pvc-tables~\cite{fink}, tuple independent tables~\cite{suciu}, BID tables~\cite{barbara}, U tables~\cite{utables} and so on, summarized by Suciu et al. in~\cite{suciu}. A common characteristic of these methods is the presence of an annotation column containing for each tuple information pertaining to the presence or provenance of the tuple in the relation, either in the form of a probability value, boolean variable, logical formula, provenance information or multiplicity. Green et al. ~\cite{green} have noted the similarity of these annotations and indicated that they are all described in a uniform approach by the semiring algebra. Dalvi et al.~\cite{dalvi,suciu} presented an important finding on the algorithmic complexity of queries in probabilistic databases, showing that they fall into two classes: some that have polynomial data complexity and others that are \#P-complete. The authors also provided a polynomial time algorithm for deciding the class of any query.

Aggregate queries constitute one of the most difficult classes of queries on probabilistic databases. Such queries on conditional tables have been considered by Lechtenborger et al. in~\cite{lechtenborger}. Amsterdamer et al.~\cite{amsterdamer} provided a theoretical framework for aggregate queries on annotated databases. Based on this work, Fink et al.~\cite{fink} developed a probabilistic database system with aggregate queries. 
Currently, the most complete probabilistic database management systems in terms of queries that they can answer are MayBMS\cite{maybms} and SPROUT~\cite{fink}. Between these, MayBMS does not provide full computation of distributions for aggregate queries, but is limited to the expected value and confidence intervals for the result. The recent work of Fink et al. with SPROUT, on the other hand, supports aggregate queries. The relational model of SPROUT, pvc-tables, combined with the query language expressed as semiring and semimodule expressions, provide a complete and uniform representation for probabilistic databases and queries. The implementation, based on decomposition trees, follows closely the theory of the semiring algebra. However, this uniform treatment of queries results in an implementation that in our opinion would be otherwise more efficient. In particular, any aggregate query results in a decomposition tree of size at least proportional to the input table, which is enormous in real applications. More importantly, they need to be kept in memory and processed. This dramatically limits the applicability of SPROUT to large scale datasets.

The extension of database management systems into the probabilistic realm is an important advancement, but the price that has to be paid for it is still too high using existing systems. The computational requirements for these systems far outweigh the deterministic counterpart. 

\section{Contributions of this paper}
\label{sec:contributions}
\ignore{\bf{Contributions of this paper}}

In this work we advance our understanding of probabilistic databases by providing a set of solutions that significantly improve the computation time for query processing reported for existing solutions. More specifically, we consider the problem of fast query evaluation in probabilistic databases. We regard a probabilistic database as a collection of probabilistic relations, for which we adopt the tuple independence model~\cite{suciu}. The result of an aggregate query in such a database is a probability distribution over the possible values of the result. 
Our performace improvement is dramatic particularly for the hardest queries, those containing aggregates. We show that it is possible to obtain practical response times to queries on probabilistic databases of up to terabyte sizes, while existing state of the art do not scale beyond gigabytes.

We maintain that the central element for an efficient implementation of a probabilistic database system should be a general, flexible data type that can be manipulated by all the standard query operators to produce the final result. The data encapsulated by this data type are in general the tuple probabilities and the aggregated value. The operators take each tuple in turn and progress through the table accumulating the final result. As we will show in the paper, our abstract data type represents a probability distribution, and its implementations allow us to express queries as operators applied to it. Moreover, different implementations of this abstract data type allow exact computation, as well as approximations\ignore{, a seamless integration of the two and transition between them.} 

We present a versatile representation for the probability distributions obtained as results of aggregate queries on probabilistic databases, based on polynomials, called probability generating function. Furthermore, we show how this representation is adapted for the key traditional aggregates, namely COUNT, SUM, MIN and MAX. For each aggregate, we develop the method that computes the exact entire distribution of the result. Given the possibly massive size of resulting distribution, we also describe some efficient methods to approximate the distributions. To integrate our mathematical apparatus for probabilistic aggregates into a database system, we also provide a mechanism for mapping probabilistic queries into deterministic query plans that can be carried out by existing database systems. To allow the mapping from probabilistic to deterministic execution
plans, we have to ensure that we can carry out the probabilistic computations required
in addition to the work that a deterministic relational database would do. To this end, we incorporate two key ingredients: an abstract data type (ADT), called PGF ADT, that encodes the exact or
approximate PGFs, and user defined aggregates (UDAs) to perform the required
operations.

We use Glade~\cite{glade} as a large scale, high performance DBMS. By doing this, we can make use of the database engine to execute the query rather than
produce a specialized execution unit. Our approach, avoiding cumbersome tree construction, maintenance and traversal, also allows us to design highly efficient operators tailored to particular situations and optimization and parallelization opportunities.

Our experiments demonstrate that with our approach we can answer TPCH queries on 1 TB databases in just one to two minutes. 

{\bf{In short, our contributions are:}}
\begin{itemize}
\item {A new representation and computational model for probabilistic databases based on probability generating functions}
\item {Approximations for aggregates that achieve a $10^-7$ precision}
\item {Exact computations for aggregates based on FFTW, }
\item {Architecture and parallelism driven implementation which allow processing billions of probabilistic tuples in practical time}
\end{itemize}

\section{Mathematical model} 
\label{sec:math}

We follow the approach in~\cite{fink} and define probabilistic tables in terms of monoid expressions. However, our goal is to directly represent probability distributions over tables as monoid expressions that can be directly evaluated by the database engine, thereby achieving uniformity of computation, leading to significantly improved computation time.

\subsection{Overview}
Suppose that we want to calculate the COUNT aggregate for the probabilistic table in Figure~\ref{ex}. We can do this by considering each tuple to be a Bernoulli random variable with probability of success given in the $p$ column. The aggregate is then given by the sum of all the Bernoulli random variables. The resulting random variable is said to have a Poisson Binomial distribution~\cite{cramer1946}. We can compute this distribution easily using probability generating functions ~\cite{goldberg}. \ignore{The probability generating function of a random variable which takes values $v_1, \dots, v_n$ with corresponding probabilities $p_1, \dots, p_n$ is the function $F(X)=p_1X^{p_1}+ \dots +p_nX^{v_n}$. An important theorem allows us to compute the PGF of a sum of independent random variables as the product of the PGFs of the components.} For our example, first we find the PGFs of the variables associated with the individual tuples. These are $F_1(X)=0.3+0.7X, F_2(X)=0.2+0.8X, F_3(X)=0.5+0.5X$. The PGF of the sum of the random variables is the product of these PGFs. Therefore the PGF giving the COUNT aggregate is given by \ignore{obtained by multiplying the previous polynomials:}$F(X)=F_1(X)F_2(X)F_3(X)=0.28 X^3+0.47 X^2+0.22 X+0.03$\ignore{, which is nothing but a simple representation of the desired distribution.} 

The SUM distribution can be computed very similarly. The Bernoulli random variables associated with each tuple are now scaled with the value in the tuple. The resulting tuple PGFs are $F_1(X)=0.3+0.7X^3, F_2(X)=0.2+0.8X^8, F_3(X)=0.5+0.5X^5$, and the sum aggregate is given by $0.28 X^{16}+0.12 X^{13}+0.28 X^{11}+0.19 X^8+0.03 X^5+0.07 X^3+0.03$. Note that, while the result of the COUNT aggregate is a polynomial, the result of the SUM aggregate in general is not, because the exponents are in general real numbers. Nevertheless, the form of the PGF resembles a polynomial, and is in fact called a generalized exponents polynomial.

Now consider what happens in the case of a MIN aggregate. The tuple random variables are the same as for the SUM aggregate. By analogy with the SUM, we would like to be able to perform a binary operation on polynomials that would give us the correct PGF as result. Noting that the values are represented in the exponents of the polynomials, we can see that this operation is multiplication, but the exponent addition becomes the MIN operation. More concretely, consider the first two tuples, having PGFs $F_1(X)=0.3+0.7X^3, F_2(X)=0.2+0.8X^8$. When both tuples are absent, with probability 0.06, the result is undefined, which we represent with the neutral element of the MIN operation, $\infty$. When the first tuple is absent and the second is present, with probability 0.24, the result is 8, which is $\min(\infty,8)$. Therefore, the correspondin term in the PGF is $0.24X^{\min (\infty,8)}=0.24X^8$. In the same way we consider all pairwise combinations of terms in the two tuple polynomials, obtaining the PGF for the two tuples: $0.06 X^\infty + 0.24X^8 +0.14X^3 + 0.56X^3 = 0.06 X^\infty + 0.24X^8 +0.7X^3$. The operation that we have to perform on PGFs is nothing but a special form of polynomial multiplication, determined by the aggregate under consideration, in this case MIN.

In classical deterministic databases, the aggregates are unified by the general theory of algebraic structures as monoids. Inspired by this observation, we ask ourselves if a similar unified treatment is possible with probabilistic databases. In the deterministic databases, the monoid elements that are aggregated are real numbers. We have seen that they translate into generalized exponents polynomials in probabilistic databases, and that the monoid operations translate into different kinds of polynomial multiplication. We formalize these observations into a precise mathematical formulation. To make our approach as general as possible, we have to allow results of aggregate queries as tuple values, that can participate in other queries. This leads to our probabilistic table model, which we also describe in this section.
 
\subsection{Generalized exponents polynomials}

We start by introducing a generalization of polynomials that allows a uniform treatment of various types of aggregation operations in relational algebra. To do so, a rigorous definition of polynomials as monoid rings is necessary.
\begin{definition}[Monoid ring~\cite{Lang}]
\label{def:monoidring}
Let $R$ be a ring and $G$ be a monoid. Consider all the functions $\phi:G\rightarrow R$ such that the set $\{g:\phi(g) \ne 0\}$ is finite. Let all such functions be element-wise addable. We define multiplication by $(\phi \psi)(g) = \sum_{k+l=g}\phi(k)\psi(l)$. The set of all such functions $\phi$, together with these two operations, forms a ring, the \emph{monoid ring} of $G$ over $R$, denoted $R[G]$.
\end{definition}

\begin{example}[The polynomial ring]
If $R$ is the ring $(\mathbb R,+,\cdot )$ and $G$ is the monoid $(\mathbb N,+)$, then $\mathbb R [\mathbb N]$ is the monoid ring of polynomials over $\mathbb R$, also denoted by $\mathbb R [X]$. A polynomial is, in terms of Definition~\ref{def:monoidring}, the function that associates the exponents to the corresponding coefficients. The multiplication rule in Definition~\ref{def:monoidring} is the usual polynomial multiplication.
\end{example}

\begin{example}[Generalized exponents polynomials]
If instead of natural numbers we choose the set of real numbers for the monoid $G$, we obtain a monoid ring similar to the polynomial ring, but the ``polynomials'' have real exponents. The definition requires that the number of terms in the polynomial is finite, because each member of the monoid ring is a function that has to be non-zero for only a finite number of its arguments. 
We can further generalize polynomials by considering $\mathbb R$ to be any monoid over the set of real numbers. For example, addition could be defined as the minimum of the two numbers. 
We call the elements of the monoid ring defined by a monoid $\mathbb R$ over the real  numbers \emph{generalized exponents polynomials}. 
\end{example}

In the rest of the paper, for brevity, we use the term \emph{polynomial} when referring to \emph{generalized exponents polynomials}. 

\subsection{Polynomial monoids}
Next, we introduce a monoid defined as set of polynomials whose coefficients add to one, together with the polynomial multiplication operation. The monoid operations translates in practice into aggregation operations performed on the probability distribution of the aggregate.\ignore{ We provide the proofs for some of the results that follow in the Appendix.}

\begin{proposition}[Polynomial monoid]
\label{thm:polymonoid}
Let $\mathbb R_*$ be a monoid over the set of real numbers and ``$\cdot$'' the multiplication operation in the polynomial monoid ring. The set of polynomials in the ring over $\mathbb R$, for which the sum of the coefficients is one, together with the operation ``$\cdot$'', forms a monoid, denoted by $\mathbb R [\mathbb R_*]$, called the \emph{polynomial monoid}.
\end{proposition}
\begin{proof}
First we have to show that the set of polynomials whose coefficients add to one is closed under multiplication. Consider two such polynomials, $A(X)=\sum_{i}p_iX^{a_i}$ and $B(X)=\sum_{j}q_jX^{b_i}$. Their product is $AB(X)=\sum_k\sum_{a_i+b_j=k}p_iq_jX^{k}$. The sum of the coefficients of $AB$ is

$\sum_k\sum_{a_i+b_j=k}p_iq_j=\sum_i\sum_jp_iq_j=\sum_ip_i\sum_jq_j=1$

Clearly, the polynomial product is associative and has a neutral element, the polynomial $1=X^0$, where 0 is the neutral element in the monoid $\mathbb R_*$
\end{proof}

\ignore{
\subsection {Probability monoid}
\label{thm:probmonoid}
\begin{proposition}[Probability monoid]
Let $\varoplus:\mathbb R \rightarrow \mathbb R$ be the operation defined by $p_1\varoplus p_2=1-(1-p_1)(1-p_2)$. The interval $[0,1]$, together with the operation $\varoplus$, forms a monoid, called the ``probability monoid''.
\end{proposition}
}

\subsection{Probability Generating Function}
Now we introduce probability generating functions as members of the polynomial monoid. They are an equivalent representation of probability distributions obtained by aggregation. We also show how aggregates can be computed as products of these functions. 
\begin{definition}[Probability Generating Function]
\label{def:ztrans}
Let $A$ be a discrete random variable taking values $a_1, \dots, a_N$ in a monoid $\mathbb R_*$.
The Probability Generating Function (PGF) of $A$ is defined as the polynomial $Q_A(X)\in \mathbb R[\mathbb R_*]$, defined by
  \begin{equation*}
    Q_A(X) = \sum_{k=1}^N P(A=a_k) X^{a_k}
  \end{equation*}
\end{definition}

\ignore{In what follows we write for short $\sum_{a}$ to denote summation over all possible values of random variable $A$.}
The following important theorem relates the computation of the PGF of a sum of random variables in any monoid to the product of the PGF of the elements. This is a generalization of Theorem 7.8 in ~\cite{goldberg} in two ways: it allows random variables to have real values, not just natural numbers, and allows the sum to be carried out in any monoid, not just normal addition.

\begin{theorem}[PGF of a sum of random variables]
\label{thm:pgf_sum}
Let $A_1,\dots, A_n$ be independent discrete random
variables taking values in a monoid $\mathbb R_*$, with corresponding PGFs $Q_{1}(X),\dots, Q_{n}(X)$, respectively. With the addition defined in $\mathbb R_*$, the PGF of $A=\sum_{i=1}^nA_i$ is 

$$Q_A(X)=\prod_{i=1}^n Q_i(X)$$
\end{theorem}
\begin{proof}
We will prove the theorem by induction. The base case is for $n=2$. Let $S=A_1+A_2$. 
Then, on one hand, by the definition of the polynomial product, we have 

$Q_1(X)Q_2(X)=\sum_{t}\sum_{a_1+a_2=t} P(A_1 \se a_1)P(A_2 \se a_2)X^t$  

On the other hand, from the definition of the PGF 

$Q_S(X)=\sum_sP(S \se s)X^s$

We have to show that for all $s$, 

$P(S \se s)=\sum_{a_1+a_2=s} P(A_1 \se a_1)P(A_2\se a_2)$. 

To do this, we write $P(S \se s)$ by conditioning on $X_2$:

$P(S \se s)=\sum_{a_2}P(S \se s|A_2 \se a_2)P(A_2 \se a_2)=\sum_{a_2}P(A_1+a_2 \se s)P(A_2 \se a_2)=\sum_{a_1+a_2=s}P(A_1 \se a_1)P(A_2 \se a_2)$

Hence, the two polynomials, $Q_S$ and $Q_1Q_2$ are equal and the base case is proved.

For the inductive case, we consider a new random variable $X_{n+1}$ independent of $X_1,\dots,X_n$. We assume that $Q(z)=\prod_{i=1}^n Q_i(z)$ and we show that the PGF of $X+X_{n+1}$ is $Q(z)Q_{n+1}(z)$. This is reduced to the base case, where the two variables are now $X$ and $X_{n+1}$. By the hypothesis, $X_{n+1}$ is independent of each $X_i$, for $i \in {1,\dots,n}$. It follows that $X_{n+1}$ is independent of $X$ (see~\cite{goldberg}, Theorem 7.8).

\end{proof}

\subsection{Probabilistic Databases}
We are now ready to assemble the theoretical concepts we have introduced into a cohesive representation for probabilistic databases, and then to define our operators acting on this representation.
\begin {definition}[Probabilistic table]
\label{def:table}
A probabilistic table $\mathcal R$ is a relation with a probability column $p\in [0,1]$ \ignore{holding values in the probability monoid $([0,1],\varoplus)$,} and where the tuple values are elements of the polynomial monoid $\mathbb R[\mathbb R_*]$. 
A probabilistic database is a set of probabilistic tables.
\end{definition}

The probability column $p$ and the distributions represented by the tuple values induce a probability distribution over relational algebra tables. Each tuple $\tau\in \mathcal R$ is present in a deterministic instance of $\mathcal R$ with probability $\tau.p$, independent of all other tuples. Once the presence of a tuple is determined, the values of the tuple in the deterministic instances are decided independently by a random process according to their distributions, represented as PGFs. 

\begin{example}
Figure~\ref{ex:pgftable} presents a simple probabilistic table. We detail the distribution over deterministic tables it describes in Figure~\ref{ex:pgftable2}.

\begin{figure}[h!]
\caption{A probabilistic table}
\centering
\begin{tabular}{c|c|c|c}
\hline
\multicolumn{1}{|c|}{R} & A & B & \multicolumn{1}{c|}{p} \\
\hline
  & $X^4$ & $0.2X^5+0.3X^7+0.5X^9$ & \multicolumn{1}{c|}{0.9} \\
  & $X^3$ & $X^2$ & \multicolumn{1}{c|}{0.2}  \\ \cline{2-4}
\end{tabular}
\label{ex:pgftable}
\end{figure}

\begin{figure}[h!]
\caption{Deterministic tables that can be instances of the probabilistic table in Figure~\ref{ex:pgftable}}
\centering
\subfloat[p=0.08]{
\begin{tabular}{c|c|c}
\hline
\multicolumn{1}{|c|}{R} & A & \multicolumn{1}{c|}{B} \\
\hline
\end{tabular}}
\qquad 
\subfloat[p=0.004]{
\begin{tabular}{c|c|c}
\hline
\multicolumn{1}{|c|}{R} & A & \multicolumn{1}{c|}{B} \\
\hline
  & \multicolumn{1}{c|}{4} & \multicolumn{1}{c|}{5} \\\cline{2-3}
\end{tabular}}
\qquad
\subfloat[p=0.06]{
\begin{tabular}{c|c|c}
\hline
\multicolumn{1}{|c|}{R} & A & \multicolumn{1}{c|}{B} \\
\hline
  & \multicolumn{1}{c|}{4} & \multicolumn{1}{c|}{7} \\\cline{2-3}
\end{tabular}}
\qquad
\subfloat[p=0.1]{
\begin{tabular}{c|c|c}
\hline
\multicolumn{1}{|c|}{R} & A & \multicolumn{1}{c|}{B} \\
\hline
  & \multicolumn{1}{c|}{4} & \multicolumn{1}{c|}{9} \\\cline{2-3}
\end{tabular}}
\qquad
\subfloat[p=0.72]{
\begin{tabular}{c|c|c}
\hline
\multicolumn{1}{|c|}{R} & A & \multicolumn{1}{c|}{B} \\
\hline
  & \multicolumn{1}{c|}{3} & \multicolumn{1}{c|}{2} \\ \cline{2-3}
\end{tabular}}
\qquad 
\subfloat[p=0.036]{
\begin{tabular}{c|c|c}
\hline
\multicolumn{1}{|c|}{R} & A & \multicolumn{1}{c|}{B} \\
\hline
  & \multicolumn{1}{c|}{4} & \multicolumn{1}{c|}{5} \\
  & \multicolumn{1}{c|}{3} & \multicolumn{1}{c|}{2} \\ \cline{2-3}
\end{tabular}}
\qquad
\subfloat[p=0.054]{
\begin{tabular}{c|c|c}
\hline
\multicolumn{1}{|c|}{R} & A & \multicolumn{1}{c|}{v} \\
\hline
  & \multicolumn{1}{c|}{4}  & \multicolumn{1}{c|}{7} \\
  & \multicolumn{1}{c|}{3}  & \multicolumn{1}{c|}{2} \\ \cline{2-3}
\end{tabular}}
\qquad
\subfloat[p=0.09]{
\begin{tabular}{c|c|c}
\hline
\multicolumn{1}{|c|}{R} & A & \multicolumn{1}{c|}{v} \\
\hline
  & \multicolumn{1}{c|}{4}  & \multicolumn{1}{c|}{9} \\
  & \multicolumn{1}{c|}{3}  & \multicolumn{1}{c|}{2} \\ \cline{2-3}
\end{tabular}}
\qquad
\label{ex:pgftable2}
\end{figure}

\end{example}

It is easy to see that there is a monoid homomorphism $\gamma$ between a monoid $\mathbb R_*$ and the polynomial monoid $\mathbb R[\mathbb R_*]$ over $\mathbb R_*$, given by $\gamma(a)=X^a$. It follows that any deterministic database can be seen as a probabilistic database through the homomorphism $\gamma$, by transforming every attribute $A$ of the deterministic database into $\gamma(A)$ and adding a probability column filled with 1, or some other available probability values. 

\subsection{Query language}
Here we present a probabilistic relational query language for probabilistic databases. First we define the aggregation operators in a uniform representation as monoid operation. 
\ignore{Let $I:\{true,false\}\rightarrow\{0,1\}$ be the indicator function, $I(true)=1; I(false)=0$.}

\begin{proposition}[Aggregation monoids]
\label{thm:aggmonoid}
 Let $+_{SUM}$, $+_{MIN}$ and $+_{MAX}$ be the functions from $\mathbb R \times \mathbb R$ to $\mathbb R$, defined in infix notation as:
\begin{align}
\ignore{a+_{COUNT}b &= I(a\ne 0) + I (b \ne 0)\\}
a+_{SUM}b &= a + b\\
a+_{MIN}b &= \min(a,b)\\
a+_{MAX}b &= \max(a,b)
\end{align}
Then the algebraic structures $\mathbb R_{SUM}=(\mathbb R,+_{SUM},0)$, $\mathbb R_{MIN}=(\mathbb R,+_{MIN},\infty)$ and $\mathbb R_{MAX}=(\mathbb R,+_{MAX},-\infty)$ are monoids.
\end{proposition}

In what follows\ignore{let $R$ be a PGF table generated randomly from the distribution given by $p$ in a probabilistic table $\mathcal R$}, let $P$ denote the probability function associated with a random variable, $F$ a polynomial in $\mathbb R[\mathbb R_*]$ and $\tilde F$ the random variable whose distribution has the PGF $F$. \ignore{Let $\phi_F$ be the polynomial $F$ seen as the function $\phi$ in definition~\ref{def:monoidring}; in other words, $\phi_F(a)$ is the coefficient of $X^a$ in $F(X)$.}Let * denote all the attributes of a relation except $p$. 

\begin{definition}
\label{algebra}
The query language $\mathcal Q$ considered in this paper consists of queries that are built using the relational operators $\sigma,\pi,\times,\Join,\varpi$. In $\pi_A$ and $\varpi_{A;\alpha_1\leftarrow(B_1),\dots,\alpha_k\leftarrow(B_k)}$, the attributes in $A$ are PGF of constants, in other words, they consists of a single term with coefficient 1. The aggregates in $\varpi$ are COUNT, SUM, MIN and MAX.\ignore{those defined by the monoids of Theorem ~\ref{thm:aggmonoid}}. The comparison operators are $=,<,\le,>,\ge$, denoted in general by $\theta$. 
\end{definition}

 The relational operators are described in what follows.

\begin{itemize}
\item {\sc projection}
The projection computes the probability of each output tuple as the probability that there is at least one tuple in the input table that projects to the output tuple. This probability can be easily computed using the independence assumption:

\begin{multline*}
\pi_{A_1,\dots, A_n}(\mathcal R) = \{(\tau.A_1,\dots,\tau.A_n,p) | \tau \in \mathcal R, \\
p=1-\prod_{\substack{\upsilon\in \mathcal R\\ \upsilon.A_1=\tau.A_1\\\dots\\\upsilon.A_n=\tau.A_n}}(1-\upsilon.p\}
\end{multline*}

\item {\sc selection}

For the selection operator we can distinguish three cases: the condition involves only deterministic attributes, the condition involves one deterministic and one probabilistic attribute, or the condition involves two probabilistic attributes. In the first case we
simply carry out the selection. Each tuple of the input relation will
be in the result with the initial probability $p$.
In the
second case, one of the attribute values 
is a distribution, represented as a PGF,
obtained from a previous aggregation operation. The PGF is truncated such that only the values that satisfy the condition are maintained as terms of the polynomials. After the illegal terms are eliminated, the PGF is re-normalized. To obtain the tuple probability, we
multiply the probability that the condition is true with the initial probability of the
tuple. In the third case, the selection operator might introduce dependencies between the two distributions, which are hard to handle. Therefore we restrict the language, imposing the condition that the attributes participating in the selection are not used in subsequent computations; they are immediately projected out after the selection. All cases can be expressed mathematically as follows:

\begin{equation*}
\begin{split}
\sigma_{A_i \theta A_j}(\mathcal R) &= \{(\tau.A_1,\dots,\tau.A_{i-1}, F'(X),\\
& \qquad \tau.A_{i+1},\dots,\tau.A_{j-1},G'(X), \\
& \qquad \tau.A_{j+1},\dots,\tau.A_n,p) | \tau\in \mathcal R(A_1,\dots,A_n), \\
& \quad \tau.A_i = F(X)=\sum_k a_kX^{v_k}, \\
& \quad \tau.A_j=G(X)=\sum_l b_lX^{w_l}, \\
& \quad F'(X) = normalize(\sum_{\substack {k\\ \exists l : a_k \theta b_l}} a_kX^{v_k}), \\
& \quad G'(X) = normalize(\sum_{\substack {l\\ \exists k : a_k \theta b_l}} b_lX^{w_l}), \\
& \quad p=\tau.p \times P(\tilde F \theta \tilde G)\}
\end{split}
\end{equation*}

\item {\sc cartesian product} 

The probability of each tuple in the result is the product of the probabilities of the participating input tuples:

\begin{equation*}
\begin{split}
\mathcal R \times \mathcal S &= \{(\tau.A_1,\dots,\tau.A_m,\tau'.B_1,\dots,\tau'.B_n,p) | \\
& \qquad \tau\in \mathcal R(A_1,\dots,A_m),\tau'\in \mathcal S(B_1,\dots,B_m), \\
& \quad p=\tau.p \times \tau'.p \}
\end{split}
\end{equation*}

\item {\sc join} 

The join is a cartesian product followed by a selection. The restriction imposed on the selection operator means that we do not allow joins in which the condition involves two probabilistic attributes.

\begin{equation*}
\mathcal R \Join_{A \theta B} \mathcal S = \sigma_{A \theta R}(\mathcal R\times \mathcal S)
\end{equation*}

\item {\sc aggregation }

We compute  aggregates following these steps:
\begin{enumerate}
\item transform each tuple value to be aggregated into the proper form required by the aggregate. This means that if the aggregate is COUNT, then each value becomes $X$. If the aggregate is SUM and the previous aggregation on that column was MIN (MAX), then the PGF is transformed such that the term $X^{\infty} (X^{-\infty})$ becomes $1=X^{0}$. If The aggregate is MIN (MAX) and the previous aggregationon the column was MAX (MIN), then $X^{-\infty}$ ($X^{\infty}$) becomes $X^{\infty} $ ($X^{-\infty}$). For an attribute $B=F(X)$, and aggregate AGG, we denote this transformation by $T_{AGG}(F(X))$. So, for example, $T_{COUNT}(F(X))=X, T_{SUM}(X^{\infty}+G(X))= X^0+G(X)$.

\item compute the PGF of each tuple as $(1-p) X^0 + p T(B)$. Here, ``0'' in $X^0$ should be understood as the neutral element of the monoid corresponding to the aggregate ($\infty$ for MIN, $-\infty$ for MAX). 

\item compute the aggregate PGF multiplying the tuple PGFs. The final result is 

\begin{equation*}
\varpi_{\emptyset;\alpha=AGG(B)}(\mathcal R)=\left\{\left(\prod ((1-p) X^0)+pT(B)),1\right)\right\}
\end{equation*}
\end{enumerate}

If multiple columns are aggregated at the same time, we are faced with a more complicated situation. The PGFs we compute as the aggregated value for each column depend on each other. Formulating this dependence mathematically is a nontrivial problem. To tackle this challenge, we eliminate the dependence between columns. More specifically, we 
define the aggregation over multiple columns as different aggregation operations over different columns in separate copies of the initial probabilistic table. We express this mathematically as follows:

$\varpi_{\emptyset;\alpha_1=AGG_1(B_1),\dots,\alpha_n=AGG_n(B_n)}(\mathcal R)=$

$\varpi_{\emptyset;\alpha_1=AGG_1(B_1)}(\mathcal R_1) \times \dots \times$ $\varpi_{\emptyset;\alpha_m=AGG(B_n)}(\mathcal R_n)$

where $\mathcal R_i$, for $i=1,\dots ,n$ are identical copies of $\mathcal R$

\item {\sc grouping with aggregation}

When grouping is performed together with aggregation, the steps are essentially the same as before, but they are performed for each group. Whereas with no aggregation the single tuple in the result is sure to exist $(p=1)$, with gouping the resulting tuples have a probability of existence depending on the probability of existence of the tuples in the input table. Mathematically, the operators can be expressed as follows

$\varpi_{A_1,\dots,A_m;\alpha=AGG(B)}(\mathcal R)=$

$\{(a_1,\dots,a_m,F,p) | $
$((a_1,\dots,a_m) \in \pi_{A_1,\dots,A_m}(\mathcal R)$,

$F(X)=\prod_{A_1=a_1,\dots,A_n=a_n} ((1-p)X^0+pT(B))$

$p=\prod_{A_1=a_1, \dots, A_n=a_n}\mathcal R.p \}$

\vspace{15pt}

$\varpi_{A_1,\dots,A_m;\alpha_1=AGG_1(B_1),\dots,\alpha_n=AGG_n(B_n)}(\mathcal R)=$ 

$\{(a_1,\dots,a_m,F_1,\dots,F_n,p)|$

$(a_1,\dots,a_m,F_1,p)\in \varpi_{A_1,\dots,A_m;\alpha_1=AGG_1(B_1)}(\mathcal R),$

$\dots$

$(a_1,\dots,a_m,F_n,p)\in \varpi_{A_1,\dots,A_m;\alpha_n=AGG_n(B_n)}(\mathcal R)\}$

where $\mathcal R_i$, for $i=1,\dots ,n$ are identical copies of $\mathcal R$

\end{itemize}

The probabilities involving comparison operators are computed as follows:


$P( \tilde A = \tilde B)=\sum_{v\in Dom(\tilde A)\cap Dom(\tilde B)} P( \tilde A = v) P( \tilde B = v)$ 

$P( \tilde A < \tilde B)=$

$\sum_{b\in Dom(\tilde B)} \sum_{\substack{a\in Dom(\tilde A)\\a<b}} P( \tilde A = a) P( \tilde B = b)$ 


\
\begin{proposition}[Closure]
\label{thm:closure}
Probabilistic tables are closed under the query language $\mathcal Q$
\end{proposition}

\begin{proof}
We prove the proposition by structural induction. Let $Q$ be any operator in $Q$. If $Q$ is selection, projection or cartesian product, clearly the result is a probabilistic table. If $Q$ is an aggregation operator, we have to show that the aggregate results are in the monoid $\mathcal R[\mathbb R_*]$; in other words, that the sum of the coefficients of the resulting polynomial is 1. This is so because the resulting polynomial is the PGF of a sum of random variables.  
\end{proof}

\begin{proposition}[Soundness]

Let $\mathcal D$ be a probabilistic database and $D$ a random instance of $\mathcal D$. Let $Q$ be a query in the language $\mathcal Q$, and $Q'$ the corresponding deterministic query. Denote by $Q(\mathcal D)$ the result of query $Q$ on the database $\mathcal D$. The probability distribution of $Q(\mathcal D)$ is equal to the probability distribution of $Q'(D)$.  
\end{proposition}
\begin{proof}
We prove the proposition by structural induction. \ignore{ First, we note that the probabilistic
tables of the base relations satisfy the result in the theorem.  Any
query expressible in the query language consists of an operator
tree. We consider possible operators at intermediate points in the
operator tree. For the induction step,}We assume that the
arguments of an operator satisfy the
result in the theorem. We will prove that the probabilistic table
computed by each operator is correct as well. \ignore{Moreover, the
probabilistic tables can only encode probabilistic spaces in which the
result tuples are independent. It is sufficient to prove that, for any
opearator, the probability that an arbitrary tuple belongs to the result is
correct. The proofs can thus focus on specific result tuples and the
way they can be produced. }

\paragraph*{Projection:} The classic approach \cite{diamonds} to
    compute the probability that a tuple belogs to the result of
    projection is to first consider the negation of this
    statement. Since the query language only allows hierarchical queries, the
    events corresponding to the existence of the tuples that project
    to the target tuple $t$ are independent. This means that, for
    tuples $t^\prime$ that poject to $t$ (denoted by $t^\prime \perp
    t$),
    \begin{equation*}
        P(t\in \mathcal R) = 1-P(t\not\in \mathcal R) = 1-\prod_{t^\prime, t^\prime\perp t} (1-P(t^\prime\in \mathcal R))
    \end{equation*}
    On the other hand, it is easy to show that \ignore{w.r.t the probability
    monoid,}
    \begin{equation*}
        p_t=\ignore{\sum_{\varoplus, t^\prime\in \mathcal R, t^\prime\perp t} p_{t^\prime} 
        =}1-\prod_{t^\prime,t^\prime\perp t} (1-P(t^\prime\in \mathcal R))
    \end{equation*}

\paragraph*{Selection (deterministic):} Deterministic selection is 
straightforward since tuples that fail the selection predicate are
eliminated and tuples that pass the predicate maintain their
probabilities.

\ignore{
\paragraph*{Selection (probabilistic):} In essence, the way the 
probabilistic table of the selection is formed is by creating a tuple
for each possible combination of the values of attributes $A$ and
$B$. In this manner, the correlation introduced by the selection
predicate $A\theta B$ being true is captured by explicitly encoding
the joint probability distribution of the combination $A,B$. If $A$
and $B$ are carried through further instead of their joint distribution
conditioned on $A\theta B$ condition, any use of $A$ and $B$ would be
incorrect. To obtain the joint distribution, we need to enumerate all
possible values $a$ of $A$ and $b$ of $B$. The probability of the
tuple $(..., a, \dots b, ...)$ is $p=\tau_p \times P(A=a) \times P(B=b)$.
}

\paragraph*{Selection (probabilistic+projection):}  \ignore{If the probabilistic 
attributes are projected out from the result of the selection, great
simplification can be obtained. This is due to the fact that}. No
correlations need to be captured since they cannot be used further due
to the elimination of the attributes. 
Since queries are hierarchical the existence probability of the tuple and the distribution of
the probabilistic aggregates are independent thus the new existence
probability of the resulting tuple is simply the product of the two.
    
\paragraph*{Cartesian product:} Let $(t,u)$ be a probabilistic tuple
in the result of the cross product. Since the query is hierarchical, the events associated with $t$ and $u$ are
independent. This includes any PGFs associated with attributes in
the two tuples. With this,
\begin{equation*}
  P( (t,u)\in \mathcal R) = P(t\in \mathcal R) \cdot P(u\in\mathcal R)
\end{equation*}
This is precisely the formula used to compute the probabilistic
table of the cross product.

\paragraph*{Join:} Follows straightforwardly form the deterministic selection and cross
   product proofs.

\paragraph*{Aggregation:} First, all the tuples that participate 
in the aggregation are required to be independent (query is hierarchical). Let $t$ be such a tuple and $B$ the attribute aggregated. As
we have seen in Section~\ref{sec:math}, the aggregate can be expressed as a
sum in a the monoid corresponding to the aggregate. By
Theorem~\ref{thm:pgf_sum}, the PGF of the sum in a monoid is the
product of PGFs as long as the parts are independent -- hierarchical queries ensure this. To encode the aggregate as a
sum of random variables, we use indicator variables that encode the
existence of the tuple multiplied by the value of the attribute. Using
the translation operator $T$, the PGF of $B$ translated in the monoid
of the aggregate is $T(B)$. The PGF of the indicator variable
multiplied by the value of $B$ is, through a simple application of
definition~\ref{def:ztrans}, $(1-p_t)+p_t T(B)$, thus the overall PGF
is $F(X)=\prod (1-p+pT(B))$

For the case when multiple aggregates are involved, separate
computations need to be performed for each aggregation. Should these
results participate together in further computation, their joint
distribution is obtained by computing the cross product. 

\paragraph*{Group By and aggregation:} The same proof as for the case without grouping applies here for each group. In addition, the probability of each tuple in the result is the probability that there is a tuple in each group of the input table.

\end{proof}
\ignore{
\begin{proposition}[\ignore{Completeness and }Succinctness]
\ignore{1. Any finite probability distribution over relational data-
  base instances can be represented by PGF databases.  2. }Given a
probabilistic database $\mathcal D$ and a fixed $\mathcal Q$ query
$Q$, the query result $Q(\mathcal D)$ can be represented as a
probabilistic table with size polynomial in the size of $\mathcal D$.
\end{proposition}
\begin{proof}
\end{proof}
}

\section{Probabilistic aggregates in practice}
\label{sec:practice}
In the vast majority of practical situations, the input PGF tables to a query will have deterministic values in the tuple attributes, not true distributions. The distributions arise as the result of aggregation operators. For the sake of the closure property, the theory we developed in Section~\ref{sec:math} is comprehensive and encompasses the general case in which queries can be applied to any type of probabilistic tables. For practical purposes, important improvements in running time can be achieved by treating the simple and most common cases separately. \ignore{The theory presented in Section~\ref{sec:theory} is very comprehensive and accomodates a wide spectrum of queries that do not often arise in practice.} Here we present some algorithms that lead to fast response time for the most common aggregate queries in probabilistic databases. These algorithms are based on the assumption that the values stored in the input tables on which the probabilistic aggregates are applied are deterministic. The input tables can be seen as traditional tables in the relational model with a special attribute $p$, representing the probability of the tuple.

\subsection{COUNT}
In what follows, let $S(A_1,... ,A_D)$ be a relation schema with $D$ attributes and $R$ an instance of that schema. 
Let $n=|R|$ and for each $t_i\in R, i\in \{1,...,n\}$, let $p_i$ be the probability of tuple $t_i$ and $q_i=1-p_i$.  

For computing the COUNT aggregate we consider each tuple in the probabilistic database to be a binary event with two possible outcomes: present or absent. Such an event is modeled by a Bernoulli random variable with probability of success $p_i$. Let $E_i$ be the random variable associated with tuple $t_i$. The PGF of $E_i$ is $Q^i_{COUNT}(X)=q_i+p_i X$.  The result of the COUNT operator is given therefore by the sum of all these variables, $C=\sum_{i=1}^n E_i$. The distribution of this summation is known as the Poisson binomial distribution~\cite{poisson1837,cramer1946}.

To compute the distribution of $C$ we apply Theorem~\ref{thm:pgf_sum} and obtain the PGF of $C$, $Q_{COUNT}(X)$:
  \begin{equation}
    Q_{COUNT}(X) = \prod_{i=1}^n ( q_i+p_i X)
  \end{equation}
When this polynomial product is expanded, the coefficients represent the distribution of the COUNT aggregate. More precisely, the coefficient of $X^k$ is the probability of the COUNT having the value $k$.
An essential part of our contribution, as we will show in the following sections, is developing efficient methods for the computation of this polynomial product.

\subsection{MIN,MAX}
\label{sec:min}
The MIN and MAX aggregates can be represented using PGF in a manner similar to the COUNT aggregate. However, as pointed out by previous work~\cite{fink,green} and in Section~\ref{sec:math}, the two monoid operations have to be interpreted in a different monoid. Specifically, the addition operation will now stand for the minimum/maximum of the two operands, and the neutral element will be $+\infty/-\infty$. In what follows we will present out approach for the MIN aggregate, as the MAX aggregate is dual to MIN.

We first present a model that expresses the distribution of MIN exactly. We call this the \emph{full distribution}. We then extend this model to express an approximation to the distribution. \ignore{As we will show later in Section~\ref{sec:exp},}This extension will further improve the performance of our method with negligible loss in accuracy.

\subsubsection{Full distribution}

Similar to the COUNT aggregate, we associate a random variable with each tuple. To compute the MIN aggregate, this random variable has to take as value the neutral element of the minimum operation if the tuple is absent. Otherwise, if the tuple is present, it has to take the value of the aggregated attribute of the tuple.

Suppose that $A_k$ is the aggregated attribute. To simplify notation, for $i\in {1,...n}$ define $a_{i}$ to be the value of the aggregated attribute in tuple $t_i$; in other words, $a_{i}=t_i|_{A_k}$. Denote by $\infty$ a value larger than all elements of the set $\{a_{1},... , a_{n}\}$.

 The PGF of the random varible associated with each tuple is

$Q^i_{MIN}(X)=q_iX^{\infty}+p_iX^{a_{i}}$.

We redefine the polynomial multiplication operation so that the addition at the exponent stands for the minimum operation. We denote the resulting operator with $\times_{MIN}$. Explicitely, if $P_1(X)=a_{0}X^{\infty}+a_{1}X+...a_{n_1}X^{n_1}$ and $P_2(X)=b_{0}X^{\infty}+b_{1}X+...b_{n_2}X^{n_2}$ then

$$P_1(X)\times_{MIN}P_2(X)=\sum_{\substack{\min(i,j)=k \\ k \le \min(n_1,n_2)}}a_ib_jX^{k}$$

The PGF for the MIN aggregate is, again, the PGF of the ``sum'' of the tuple random variables, or the ``product'' of the tuple PGF's, where by ``sum'' and ``product'' we mean the newly defined operations.

\begin{equation*}
\begin{split}
Q_{MIN}(X) &= Q^1_{MIN}(X) \times_{MIN} ... \times_{MIN} Q^n_{MIN}(X) \\ 
&=  (q_1 X^\infty + p_1 X^a_1)\times_{MIN}... \times_{MIN}(q_n X^\infty + p_n X^a_n) \\
&=  \sum_{\alpha_i\in \{a_{1},...,a_{n}, \infty\}} \prod_{a_{j}<\alpha_i} q_j \left(1- \prod_{a_{j}=\alpha_i} q_j\right)X^{\alpha_{i}}.
\end{split}
\end{equation*}
 
We note that $P(\min_{i=1,\dots,n}\{a_{i}\}=M)$ is the coefficient of $X^M$ in $Q_{MIN}(X)$.

\subsubsection{Approximation}
As the value of the distribution of MIN increases, the tuple probabilities being equal, the corresponding probability quickly decreases. This is because all the records with smaller values have to be absent. The probability of this event is the product of the probabilities of their absence. For example, consider a probabilistic relation which contains one copy of each of the integers between 1 and 10000 for some attribute, each tuple having a probability of 0.5. The probability of the MIN function returning the value 50 is $1/2^{50}\approx 10^{-15}$. As the value of the MIN function increases, its probability gets even smaller. We see that, while the full distribution of the MIN aggregate has 10000 entries, most of them are irrelevant. In consequence, we have implemented a variant of the algorithm that limits the computation to the smallest $\kappa$ possible values of the minimum. 

\subsection{SUM}
The SUM aggregate is the most challenging to compute in large probabilistic databases, due to the exponential number of possible result values for a probabilistic table with $n$ tuples. Our general approach using the PGF works in this case as well. However, it is impractical for large databases. In this situation we resort to approximate solutions. Below, we first present the exact theoretical solution for the general case, when the SUM attribute can take any real values. We then describe the exact solution for an interesting special case when the SUM attribute can only take integer values. Finally, we develop two approximation strategies, namely normal and moment-based approximations.

\subsubsection{General case}
For this aggregate, we associate a random variable $S_i$ with each tuple, which takes value 0 if the tuple is not present and value $a_{i}$ if the tuple is present. In other words, $S_i=a_{i}E_i$. The PGF of $S_i$ is  
$$Q^i_{SUM}(X)=q_i+p_iz^{a_{i}}=Q^i_{COUNT}(X^{a_{i}})$$

We obtain the PGF of the SUM aggregate by applying Theorem~\ref{thm:pgf_sum}:

$$Q_{SUM}(X)=\prod_{i=1}^{n}\left(q_i+p_iX^{a_{i}}\right)$$

Assuming that each value may appear several times in the relation, we group the tuples by the value of $a_i$ and compute COUNT aggregates on the groups. We denote the COUNT PGFs of the groups by $Q_{COUNT}^{a_i}$. If there are $d$ distinct terms $\{\alpha_1, \dots, \alpha_d\}$, the PGF of the SUM aggregate can be written as:

$$Q_{SUM}(X)=\prod_{k=1}^d\prod_{a_i=\alpha_k}\left(q_i+p_iX^{\alpha_k}\right)=\prod_{k=1}^dQ_{COUNT}^k(X^{\alpha_k})$$
\subsubsection{Limited case}
The preceding equation holds for the most general case when the values $a_i$ are real numbers. In this case, the computation is exponential in the number of tuples. For practial situations in which large tables are involved, we limit our description to the case when the possible values for the attribute being aggregated is in the set $\{0,1,...,m\}$. This allows algorithmic optimizations for the exact computation using the COUNT aggregate and FFT, as we describe in Section ~\ref{sec:imp}. It is worth noting that rational numbers with a fixed number of decimal places (i.e., sets of numbers of the form \ignore{$x, x+y, x+2y,\dots,x+dy$, where $x$ and $y$ are arbitrary real numbers and $d$ is a positive integer}$a\cdot 10^{-b}$, where $a$ and $b$ are arbitrary natural numbers) can be dealt with by appropriate scaling. This restriction of the probabilistic SUM aggregate for polynomial time computation has been considered in the literature before~\cite{fink}.

\subsubsection{Approximations}
For large tables, even if we limit the aggregated values to nonnegative integers, the full sum distribution will be enormous. This is because the sum of a subset of tuples can take many different values. A user will hardly ever need a distribution table of billions of entries, and the presentation of such a distribution would in itself be a challenge. It is more likely that the queries will be limited to range queries, or the user needs an idea of how the probability mass of the aggregate is distributed. For such typical scenarios, approximations of the distribution that can be computed quickly are more desirable. Next, we present two alternative ways to approximate this distribution.

{\bf{Normal}}

We approximate the true distribution with the normal distribution. 
For tuples with attribute values $v_i$ and probabilities $p_i$ we compute the mean $\mu=\sum_iv_ip_i$ and variance $\sigma^2 = \sum_iv_i^2p_i - \mu^2$ of the probabilistic sum aggregate and we approximate the true distribution by the normal distribution with these parameters. We consider the probability mass at a particular point $s$ to be given by the integration of the pdf around $s$:

$p(S=s)=\int_{s-1/2}^{s+1/2}f(x)dx$, where $f(x)=\mathcal{N}(x;\mu,\sigma)$

{\bf{Moment-based}}
We implemented the moment-based approximation described by Lindsay~\cite{lindsay2000,lindsay89}. This method approximates an unkown distribution with a mixture of $p$ distributions from the same family, in particular gamma distributions. In order to do this, the exact moments of the unknown distribution, up to order $2p$, have to be known. The gamma distribution is given by 
\begin{equation}
f(x;\alpha,\mu)=
\frac{1}{\Gamma\left(\alpha\right)}
\left(\frac{\alpha}{\mu}\right)^\alpha
x^{\alpha-1}
e^{-\frac{\alpha x}{\mu}}  
\end{equation}

The approximating distribution is a mixture of $p$ such gamma distributions, with mixing parameters $\pi_j$, a common dispersion parameter $\lambda=1/\alpha$ and mean parameters $\mu_j$:  
\begin{equation}
g(x;\lambda, \mathbf \mu) = \sum_{j=1}^p \pi_j f(x;\frac{1}{\lambda},\mu_j)
\end{equation}

The moments method consists in finding the unknown parameters $\pi_j$, $\lambda$ and $\mu_j$ such that the first $2p$ moments of $g$ match exactly the true moments of the unknown distribution. Let us denote the moments of the true distribution by $m_1, m_2, \dots, m_{2p}$. 
We define
\begin{equation}
\delta^*_r(\lambda)=\frac{m_r}{(1+\lambda)(1+2\lambda)\dots (1+(r-1)\lambda)}
\end{equation}

and the pseudo-moment matrix 

\begin{equation}
\Delta_p(\lambda)=\{ \delta^*_{j+k} (\lambda) \} _{\substack{j=0\dots p\\k=0\dots p}}
\end{equation}

We compute the values of the parameters as follows:
\begin{enumerate}
\item Find $\lambda_1$, the solution of $\det \Delta_1\left(\lambda_1\right)=0$.
For $k=2\dots p$, find $\lambda_k$, the solution of $\det \Delta_k(\lambda_k)=0$ in the interval $[0,\lambda_{k-1})$. The solutions are unique and guaranteed to exist. Set $\lambda=\lambda_p$.

\item 
Compute $\Delta_p(\lambda)^{-1}$. This can be done by the eigenvalue decomposition method: $\Delta_p(\lambda)=ADA^T$ and $\Delta_p(\lambda)^{-1}=AD^{-1}A^T$. Since $D$ is a diagonal matrix, the diagonal elements of $D^{-1}$ are the inverses of the corresponding diagonal elements of $D$. Next, we solve the polynomial equation whose coefficients are the last row of $\Delta_p(\lambda)^{-1}$. The roots are the parameters $\mu_1, \mu_2, \dots \mu_p$.

\item Find the mixing coefficients $\pi_1,\pi_2,\dots,\pi_p$ as the solution to the system of equations
\begin{equation}
 \begin{pmatrix}
  1 & 1 & \cdots & 1 \\
  \mu_{1} & \mu_{2} & \cdots & \mu_{p} \\
  \vdots  & \vdots  & \ddots & \vdots  \\
  \mu_{1}^{p-1} & \mu_{2}^{p-1} & \cdots & \mu_{p}^{p-1}
 \end{pmatrix}
 \begin{pmatrix}
  \pi_1 \\
  \pi_1 \\
  \vdots  \\
  \pi_p 
 \end{pmatrix}=
 \begin{pmatrix}
  1 \\
  \delta_1^*(\lambda) \\
  \vdots  \\
  \delta_{p-1}^*(\lambda) \\
 \end{pmatrix}
 \end{equation}
\end{enumerate}

We compute the true moments of the sum distribution from the cumulants, using the standard formulas. We compute the cumulants as follows. Suppose that the values in the database are $v_1,\dots, v_n$ and the corresponding probabilities are $p_1,\dots, p_n$. For each tuple $i$ we define a Bernoulli random variable $A_i$ with probability of success $p_i$. The sum is then described by the random variable $A=\sum_{i=1}^nv_iA_i$. The $j$-th cumulant is $\kappa_j(A)=\sum_{i=1}^nv_i\kappa_j(A_i)$. The first cumulant of the Bernoulli random variables is $p_i$. The rest of the cumulants are computed recursively using the formula $\kappa_{j+1}=p_i(1-p_i)d\kappa_j/dp$.
 
To avoid numerical instability problems due to limited machine precision, it is crucial to transform the distribution such that the moments are computationally manageable. We apply the transformation $Z=(A-\mu + 10\sigma)/\sigma=(A-\kappa_1(A))/\sqrt{\kappa_2(A)}+10$. Deviating ten standard deviations from the mean of the original distribution allows us to fit a gamma mixture to the transformed distribution and also allows us to apply the method for sums that include negative values.

\section{From Probabilistic to Deterministic Plans}
\label{sec:queries}
In this section, we explain how we map an execution plan involving
probabilistic operators into an equivalent deterministic plan that carries out the computation. 

 In what follows, we first define the PGF from two perspectives: as a user defined aggregate (UDA) and as an abstract data type (ADT). Then we provide the translation from probabilistic to deterministic
operations.
We defer the details of various implementations of the PGF to Section~\ref{sec:imp}. \ignore{We use C++ like
code to specify the ADT.}

\subsection{The PGF as a user defined aggregate}

The result of aggregates in probabilistic databases is a probability distribution. We consider it to be a PGF. Mathematical operations on the various monoids presented in Section~\ref{sec:math} correspond to types of aggregates. The details of the computation of the result depend on the monoid/aggregate, but the structure of the computation is the same. This structure is best described as a \emph{user defined aggregate (UDA)}~\cite{sql3}). The PGF UDA has to support operations needed to compute
the aggregates in Section~\ref{sec:math}. The aggregation operations are implementations of the PGF ADT. Among those operations,the MIN/MAX aggregate accomodates an exact computation, as well as an approximation that omits computing the probabilities for very unlikely values. The SUM aggregate has an exact implementation, as well as two approximate implementations, normal and moment-based. 
The interface of the PGF UDA is defined by the standard UDA functions \texttt{Initialize, Accumulate, Merge, Finalize}. \ignore{shown in Figure~\ref{alg:pgf}.

\begin{figure}
\caption{The PGF UDA}
\label{alg:pgf}
{\scriptsize
\lstinputlisting[language=c++, frame=single]{PGFUDA.h}
}
\end{figure}
}

The system calls the \texttt{Initialize} method when the query starts executing and initializes an empty probability distribution. The DBMS scans the input tables and for each tuple it calls the \texttt{Accumulate} method, which incorporates the new tuple into the current distribution. It can create and run several PGFs for each query on different parts of the data, possibly in parallel, depending on the availability of processing units. The \texttt{Merge} method merges the results of two PGFs belonging to the same query. In order to minimize the computation cost and space usage, we do not keep the PGFs in expanded form until the end of the table scan. The system signals the actual computation of the distribution when it calls the method \texttt{Finalize}. The UDAs map one to one to aggregation monoids. The addition in each monoid is implemented by the \texttt{Accumulate} method of the UDAs.

\subsection{The PGF abstract data type}

In order to use the PGF as a value in our algebra, we need to be able to use it in comparison operations. We do this by providing an interface to the PGF that allows computing the probability that the aggregate represented by the PGF is greater than and equal to some other value, which could be a scalar or another PGF (Figure ~\ref{alg:pgf}). This interface defines the PGF abstract data type (ADT). Appropriate functions of the PGF ADT implement the comparison operations. In addition, the ADT provides a function for obtaining a confidence interval such that the probability of the aggregate being in that interval is a specified value.

\begin{figure}
\caption{The PGF ADT}
\label{alg:pgf}
{\scriptsize
\begin{lstlisting}[language=c++, frame=single]
  Iterator distribution;
  PGF(); 
  float Equal(float a);
  float Greater(float a);
  float GreaterEq(float a);
  float Equal(PGF2& other);
  float Greater(PGF2& other);
  float GreaterEq(PGF2& other);
  {low,high} ConfidenceInterval(float conf);
\end{lstlisting}
}
\end{figure}

\subsection{Probabilistic to Deterministic Operator Mapping}

We now show how we use the PGF abstract data type together with user defined aggregates (UDAs) to translate from probabilistic operators to
deterministic operators in extended relational algebra. For any
relation in the probabilistic query, say $R$, we designate by $R^p$ the
\emph{probabilistically enhanced} relation, i.e. the relation with the
extra attribute $p$, the tuple probability. In practice, $R^p$ is a deterministic relation obtained from an existing deterministic relation by adding the probability attribute. One way to do this is to create this attribute as a computed column, by applying some function to the existing attributes.\ignore{ We introduce an extended relational algebra to express probabilistic operators. In this algebra,} In our implementation of the PGF tables presented in Section~\ref{sec:math} there are two kinds of column types: single valued and probability distributions. Single valued columns are those that can only hold scalar values, in conformity to 1NF. The probability distributions are represented as PGF objects and are obtained by applying probabilistic aggregation operators to single valued columns. 

We provide the translation from probabilistic operators in the language $\mathcal Q$ to extended relational algebra operators in
Table~\ref{tbl:trans}. In this table $\pi$ is the extended projection operator that
can synthesize columns and performs no duplicate elimination, and $\varpi$
is the aggregation operator. The star denotes all attributes except the probability attribute for the probabilistically enhanced relations. We elaborate on each row of this table next.

\begin{table*}
\begin{tabular}{llll}
  Id & Probabilistic & Deterministic & Comments \\
  \hline \\
  I & $R$ & $R^p=\pi_{*,p=\mathcal{P}()}(R)$ & compute probability \\
  II & $\sigma^p_C(R)$ & $\sigma_C(R^p)$ & C: deterministic condition \\
  III & $\sigma^p_{A\theta B}(R)$ & $\pi_{*,p=p\times A.\theta(B)}(R^p) )$ & $A\theta B$: probabilistic cond.\\
  IV & $R \Join_C^p S$ & $\pi_{*,p=R.p\times S.p}( R^p \Join_C S^p )$ & $C$ is deterministic \\
  V & $\pi^p_{A_1,\dots, A_n}(R)$ & $\varpi_{A1,\dots,A_n; p=\mathcal{A}(p)}(R^p)$ & $\mathcal{A}$ is the \texttt{AtLeastOne} UDA \\
  VI & $\varpi^p_{A_1,\dots;\alpha_1=AGG_1(B_1),\dots}(R) $ & $\varpi_{A_1,\dots;p=\mathcal{A}(p), \alpha_1=PGF_1(p,B_1),\dots}(R^p)$ & $A_1,\dots,A_k,B_1,\dots,B_k$ deterministic\\
\end{tabular}
\caption{Conversion from probabilistic to deterministic query plans}\label{tbl:trans}
\end{table*}

{I)} This row does not describe a regular operator, but expresses how a conceptual probabilistic relation materializes as a deterministic relation by introducing an extra probability attribute for each tuple. 

{II)} For the probabilistic selection operator, we need to distinguish between two cases: the condition involves only single-valued columns (second row), or it involves aggregation columns (third row). In the first case we call the condition deterministic, and the database system simply carries out the selection on the probabilistically enhanced relation. Each tuple of the input relation will be in the result with the initial probability $p$. In the second case, the operands of the condition operator $\theta$ might each be a distribution, represented as a PGF UDA, obtained from a previous aggregation operation. In this case 
we assume that the probabilistic attributes involved in the condition are discarded after the selection (they are projected out). Then the selection can be written as 
 $\sigma_{A_i \theta A_j}(\mathcal R) = $
$\{(\tau.A_1,\dots,\tau.A_n,p) | $
$\tau\in \mathcal R(A_1,\dots,A_n),$
$p=\tau.p \times P(\tilde A_i=\tilde A_j)\}$.
In other words, the probability column will have the probability that the selection condition is true for the input tuple. To obtain this value, we call the operation corresponding to $\theta$ on the PGF ADT, which returns the probability that the condition is true. We multiply this probability with the initial probability of the tuple in order to obtain the final tuple probability. 

{IV)} The join operation joins the tuples in the probabilistically enhanced relation according to the join condition $C$ and adds the probability attribute, which is the product of the probabilities of the participant tuples. We assume that the condition for the join operator is deterministic, i.e. the columns that participate in the condition are single valued, not aggregates. This requirement is not a restriction, because we can rewrite a join with a probabilistic condition as a join without a probabilistic condition as $R \Join_C^p S = \sigma_C(R \Join_{true}^p S)$. However, we do not allow joins where both participating attributes are probabilistic. \ignore{We do not discuss the cross product operator because we can express it in a similar manner using the join operator.} 

{V)} The probabilistic projection operator simply applies a GROUP BY operation on the probabilistically enhanced relation. The probability of each tuple in the result is the probability that at least one of the tuples of each group is present in the input table. This is computed by a UDA called \texttt{AtLeastOne}, denoted by $\mathcal A$\ignore{, shown in Figure~\ref{alg:atleastone}}, which computes that at least one tuple in the probabilistic table is present in each group. \ignore{implements the sum in the probability ring over all tuple probabilities in each group. }Here we require that the columns being projected on are single-valued.

 VI) Finally, the aggregation operators use the PGF ADT to produce the distribution, which we store in the output table as a tuple value. Depending on the type of aggregation, we create an appropriate PGF object for each group. Notice that the PGF ADT is at the same time a UDA. For each tuple of the group, the database system calls the \texttt{Accumulate} method, which incorporates the new value into the current probability distribution. \ignore{The system decides if an exact implementation or an approximation is more appropriate for the situation and might even switch from exact computation to approximation depending on current workload.} The system stores the final PGF as a value in the aggregated column and computes the tuple probability similar to the projection operator.

\section{Implementation details}
\label{sec:imp}

In this section we delve more deeply into the implementation details of our method. We implemented our system in C++, as an extension to Glade~\cite{glade}. We explore each aggregate operator in a separate subsection.

\subsection{Comparison operators}
The comparison operators are described by the PGF ADT interface presented in Section~\ref{sec:queries} and are implemented following the math presented in Section~\ref{sec:math}. For computing the probability that the random variable represented by current PGF is equal to a constant value, we just report the probability mass at that value. For computing the probability that the random variable represented by the current PGF is equal to a random variable represented by another PGF, we iterate through the values of the current PGF and accumulate the product of the probabilities of both PGFs at the current value, assuming independence of the two variables. For computing the probability that the random variable represented by the current PGF is less than a constant value, we iterate through the values of the current PGF until that constant value and accumulate the probabilities. For computing the probability that the random variable represented by the current PGF is equal to a random variable represented by another PGF, we iterate through the values of the other PGF and accumulate the probability that the current PGF is less than the current value.

\subsection{COUNT}

For the COUNT aggregate we provide two implementations: exact and approximate. For the exact implementation, we need to carry out the polynomial multiplication $\prod_{i=1}^N ( q_i+p_i z)$. Our method is to multiply all the one-degree polynomials two by two, then the 2-degree results two by two and so on, until we get a single polynomial as the final result. We use the FFT method~\cite{fft:ct,fft:gauss} for fast polynomial multiplication, which has an $O(n \log n)$ time complexity, for an overall complexity for the full distribution of $O(n \log^2 n)$. We use the FFTW implementation~\cite{fftw} of FFT. The classical polynomial multiplication algorithm has a time complexity of $O(n^2)$. For small degree polynomials, i.e., polynomials of degree less than 5000, the overhead of Fourier transform operations is more than the trivial cost of the $O(n^2)$ algorithm. We use the classical $O(n^2)$ method for polynomials of degree smaller than this threshold and the $O(n \log ^2 n)$ algorithm for larger polynomials.

For the approximate implementation we use the algorithm for the SUM approximation, described in Section~\ref{sec:math}, with the values for all tuples set to 1.

\subsection{MIN/MAX}
The minimum of an attribute $A$ is equal to a given value $M$ if all the tuples that have in attribute $A$ values smaller than $M$ are absent and at least one tuple that has value $M$ in attribute $A$ is present. Following from the tuple independence property, the probability of this event can be factorized as presented in Section~\ref{sec:min}. The PGF ADT implementation that computes this formula uses the \texttt{AtLeastOne} UDA. We maintain an ordered list of (\texttt{value, AtLeastOne}) pairs. The \texttt{AtLeastOne} object associated to a value computes the probability that at least one tuple with that value is present. By subtracting this probability from unity we obtain the probability that no tuple is present. When a new tuple is examined in \texttt{Accumulate}, we check if the value has been seen before. If it has not, a new \texttt{AtLeastOne} UDA is created and a new pair is inserted in the list. If it has, the entry in the list corresponding to the value is identified. In both cases, the corresponding \texttt{AtLeastOne} object in the ordered list is updated by calling its \texttt{Accumulate} method. The \texttt{Merge} method goes through the entries of the other list and combines them with the list of this object, calling the \texttt{AtLeastOne.Merge} method for matching values and creating new entries for new values. To compute the probability that the minimum is equal to some value denoted by $a$, the \texttt{Equal} method iterates through the ordered list and multiplies together the results of the \texttt{AtLeastOne} objects while the list values are smaller than $a$. Finally, the result corresponding to $a$ is multiplied into the product. \ignore{We present in Figure~\ref{alg:min} part of the code for computing the entire distribution for the aggregate. }To compute the approximation that neglects the unlikely large values, some extra steps are taken to limit the number of values that are maintained by the PGF UDA to some predetermined capacity: whenever a new value is added to the list, we check if we exceeded the capacity, and if so, we eliminate the largest value. 

\subsection{SUM}

The SUM aggregate implementation is similar to MIN. We maintain an ordered list of (value, count-PGF) pairs. The PGF associated to its value is the distribution for the count of the tuples with that value. When a new tuple is examined in \texttt{Accumulate}, we search the list and create a new entry for the value if one doesn't exist yet. The corresponding PGF in the ordered list is updated by calling its \texttt{Accumulate} method. When all the tuples have been seen, the final sum PGF is computed using the count PGFs, taking advantage of the fast polynomial multiplication algorithm as seen in Section~\ref{sec:math}. After the partial count PGFs are computed, each is ``evaluated'' at $z^{\alpha_i}$ by spreading the coefficients $\alpha_i$ positions apart, inserting zeros in the empty spaces. For example, for list item $(3, 0.2z^2+0.3z+0.5)$ we create the polynomial $0.2z^6+0z^5+0z^4+0.3z^3+0z^2+0z+0.5$. In the end, we use FFTW again to multiply the extended polynomials and thus obtain the final result. \ignore{The algorithm is presented in part in Figure~\ref{alg:sum}.}

For the moment-based approximation the implementation follows the strategy outlined in Section~\ref{sec:math}. In the scanning phase the method \texttt{Accumulate} just accumulates the cumulants of the true distribution. The rest of the work is done by the \texttt{Finalize} method. Before using the moments, in order to avoid numerical instability problems, it is crucial to scale the moments as described in Section \ref{sec:practice}. The moment based approximation provides a very accurate estimate of the cdf of the distribution. As with any approximation, the pdf is more problematic. Some values in the covered range might not be possible at all, but the approximation will still have a non-zero value. \ignore{In Figure~\ref{alg:kolmogorov} we present the pseudocode for our implementation. We only show the \texttt{Greater} method, which computes the cdf using the GSL library. The \texttt{EqualTo} method returns the probability mass in an interval of length $1$ around the argument.} To compute the probability that the sum is equal to a specific value we return the probability mass in an interval of length $1$ around the value. We use GSL~\cite{gsl} for numerical solutions involved in our method.

\section{Experimental evaluation}
\label{sec:exp}

As stated in the introduction, the main goal of this paper is to show
that it is possible to run probabilistic queries efficiently at
terabyte database size. In this section we provide convincing evidence that the
goal has been achieved. The specific goals of the empirical evaluation
are:
\begin{compactitem}
\item show that ideas in this paper can be incorporated into a high-performance database system
\item asses the performance penalty probabilistic queries incur for both exact and approximate methods
\item asses the accuracy of the approximation methods when compared to the exact methods
\end{compactitem}
As we will show, when the approximate methods described in Section~\ref{sec:practice} for aggregates are used, the performance is on par with the deterministic queries. At the same time, the approximation accuracy is truly exceptional: $10^{-6}$ relative error or less. Given the large scale of our experiments -- 1TB scale -- 
these results support our claim that probabilistic queries can be tackled for large data.

\subsection{Experimental Setup} 

\paragraph*{Dataset:} All experiments we perform are based on the TPC-H benchmark\cite{tpch} at the scale factor 1000 -- this corresponds to a total storage of 1TB. The largest relation is \texttt{lineitem} and has 6 billion tuples at this scale factor, next largest, \texttt{orders}, has 1.5 billion tuples. For each of the relations, an extra column, \texttt{p}, with a randomly selected number between 0.0 and 1.0 is added and encodes the probability that the corresponding tuple exists. All probabilistic queries involve this column for all relations involved. 

It is worth noting that the experiments in this paper involve data much larger than any published results on probabilistic queries: 1TB vs the 1GB used by non-aggregate queries by \cite{olteanu2009ICDE,fink,wick}, the previous largest dataset.

\paragraph*{Queries:} The queries used are variations of 17 out of 22 TPC-H queries --  all queries that do not contain \texttt{EXCEPT} clauses are included. For each TPCH query we executed four versions: 
\begin{compactitem}
\item the computation of the confidence of obtaining at least one tuple in the result, which corresponds to the boolean versions of the queries in~\cite{olteanu2009ICDE} (labeled ``confidence''); 
\item the confidence computation for each group in the result set, which corresponds to the non-boolean version of the queries in~\cite{olteanu2009ICDE} (labeled ``group confidence''); 
\item the probabilistic TPCH query with aggregates (labeled ``aggregates'') and
\item the original, deterministic TPCH query (labeled ``deterministic''). 
\end{compactitem}
While due to lack of space we cannot list the specific probabilistic versions of the TPC-H queries that we used, we exemplify the more complex version, \emph{aggregate}, by showing the probabilistic version of Q20 in Figure~\ref{fig:query}.

\begin{figure}
\begin{align*}
R_1 &= \sigma_\text{p\_name LIKE ``forest\%''}(part)\\
R_2 &= \text{partsupp} \Join_\text{p\_partkey=ps\_partkey} R_1\\
R_3 &= \sigma_\text{l\_shipdate BETWEEN (1995/01/01, 1996/01/01)}(\text{lineitem})\\
R_4 &= \text{partsupp} \Join_\text{l\_partkey=ps\_partkey AND l\_suppkey=ps\_suppkey} R_3\\
R_5 &= R_1 \Join_\text{l\_suppkey=ps\_suppkey} R_4\\
R_6 &= \varpi_{\text{ps\_suppkey},\alpha=\text{SUM(l\_quantity)}}(R_5)\\
R_7 &= \pi_{*,p = 1 - \alpha.\text{Greater(ps\_availqty)}}(R_6)\\
R_8 &= \sigma_\text{n\_name = ``CANADA''}(\text{nation})\\
R_9 &= \text{supplier} \Join_\text{s\_nationkey=n\_nationkey} R_8\\
R_{10} &= R_7 \Join_\text{s\_suppkey=ps\_suppkey} R_9\\
Q_{20} &= \pi_{\text{s\_name,s\_address}} R_{10}
\end{align*}
\caption{Probabilistic version of TPCH query 20}
\label{fig:query}
\end{figure}

\paragraph*{Computer system and implementation:} For all the
experiments in this paper we use a mid-range server with 4 AMD Opteron
6168 processors, each with 12 cores (48 cores total) running at
1.9GHz, 256GB main memory and 76 hard drives attached through 3 LSI
host addapters. The system is capable of streaming data from the disk
array at rates of up to 3GB/s. A system with similar capabilities can
be purchased at the time of writing this paper for less than 20,000\$.

The probabilistic monoids both for the exact and probabilistic version
are encoded as combined ADT and UDFs as explained in Section~\ref{sec:queries}
in Glade/DataPath\cite{glade,datapath}, a high-performance relational
database system that has extensive support for advanced user defined
types (ADT) and user defined aggregates. The actual encoding is in
C++; queries are expressed in Piggy, a Pig Latin like dataflow
language supported by Glade.

\subsection{Large Scale Experiments}

For our first set of experiments, we execute 17 out of 22 TPC-H
queries at 1TB scale using the 4 query variants: deterministic,
confidence, group confidence and aggregate. The aggregate queries use
the moment based approximation in Section\ref{sec:practice}.  For comparison
purposes, we include in the result published performance on these
queries for the 10th best TPC-H result -- a 2M\$ system running Oracle
11g on 64 core SPARC64 VII+, with 512GB memory, 80 24GB flash
hard drives. Results are depicted in Figure~\ref{fig:ICDE09}. A number
of interesting conclusions can be drawn:
\begin{inparaenum}
\item the execution of aggregate probabilistic queries is not
  significantly slower than deterministic queries
\item all queries run faster than 2 minutes at 1TB scale
\item on some of the harder queries like Q1, Q18 and Q21, the
  probabilistic aggregate query is faster than Oracle on the deterministic query.
\end{inparaenum}
As a secondary comparison, since a direct comparison at 1TB scale is
not possible with the current MayBMS implementation (which is based on
Postgres, which does not scale to such large scale), we present in
Figure~\ref{fig:ICDE09comp} the results reported in
\cite{olteanu2009ICDE} for 1GB scale TPC-H. Notice that, even under
ideal scaling, the group confidence versions of Q3 and Q18 would take
10,000s, at least 100 times longer than our solution. Only
non-aggregate queries were included in Figure~\ref{fig:ICDE09comp}
since MayBMS does not allow aggregate queries.

As it is clearly evident from these experimental results,
probabilistic queries on terabyte databases are feasible with the
techniques introduced in this paper, as long as the approximate methods
for aggregate computation are used. To address concerns about the
accuracy and the need for this approximations, we further study them
in the remainder of this section.

\subsection{Approximate vs. Exact Aggregates}

In its most general form, the SUM aggregates require exponential
effort (the problem is noted to be NP hard in \cite{fink}). This means
that even tables with 100 tuples are out of reach of minute long
computation. For this reason, we focus on the simplest version of
aggregate queries, COUNT, and study the performance based on a
COUNT(*) aggregate over a filtered version of \texttt{lineitem}. The
filtering predicate will allow us to select between 100M and 1 billion
tuples for the COUNT aggregation. The running time experiments and
comparison with deterministic and approximate count based on moments
are depicted in Figure~\ref{fig:count}. We also included in the
comparison the running time reported in \cite{fink} for this query at
scale factor 0.1 (600,000 tuples in \texttt{lineitem}). These
experiments reveal that, as seen before, the approximate method is
within a small factor of the deterministic method and between 7.6-37
times faster than the exact computation. Also notice that our exact
computation is much faster than the one in \cite{fink} since problems
100 times larger can be accommodated in the same time -- this is mostly
due to the use of parallelism and FFTW.

The above experiments revealed that the approximate aggregation is much
faster than the exact computation for COUNT. This gap widens dramatically
for SUM. While slower, exact computation allows us to asses the precision of
the approximate aggregation. The result for the same experiment as
above is depicted in Figure~\ref{fig:error}, where the relative error
of the computation of the lower end of the .95 confidence interval is
depicted. These results reveal that the error is really small (between
$3\cdot 10^{-7}$ and $2\cdot 10^{-9}$) thus there is no practical
reason to prefer the exact computation.

\begin{figure*}[t]
\centering
\includegraphics[scale=1]{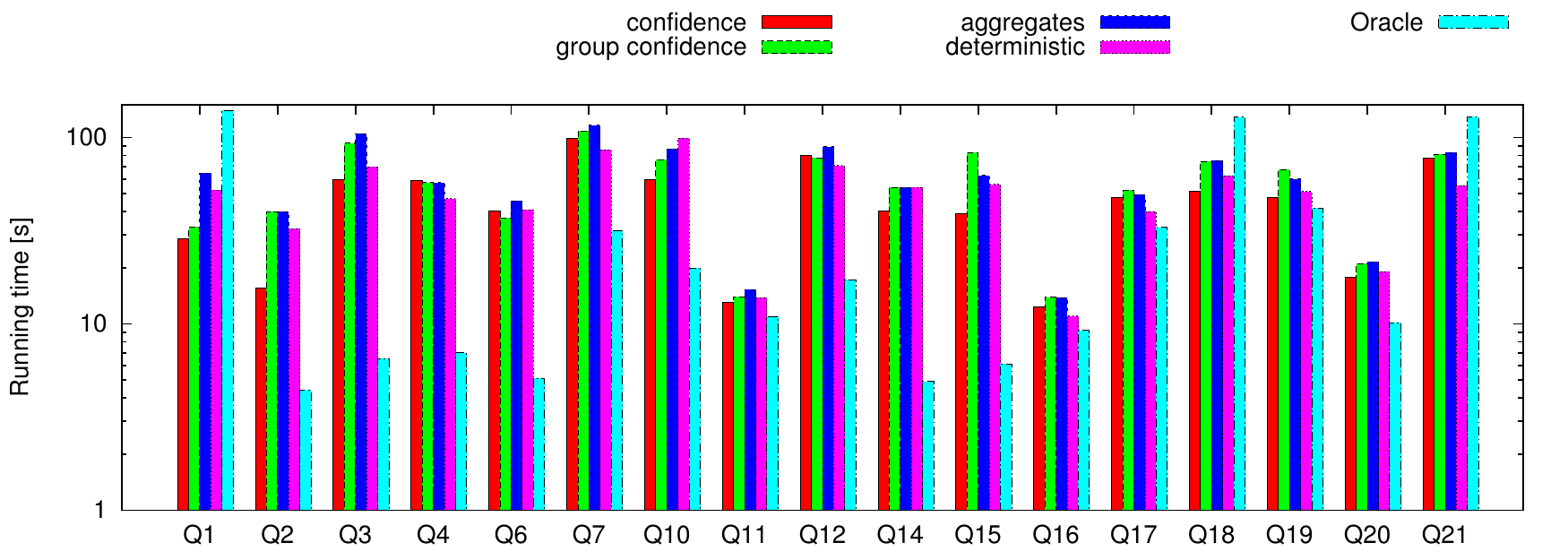}
\caption {Running time for our system on TPCH queries, 1TB database size} 
\label{fig:ICDE09}
\end{figure*}

\begin{figure*}[t]
\centering
\includegraphics[scale=1]{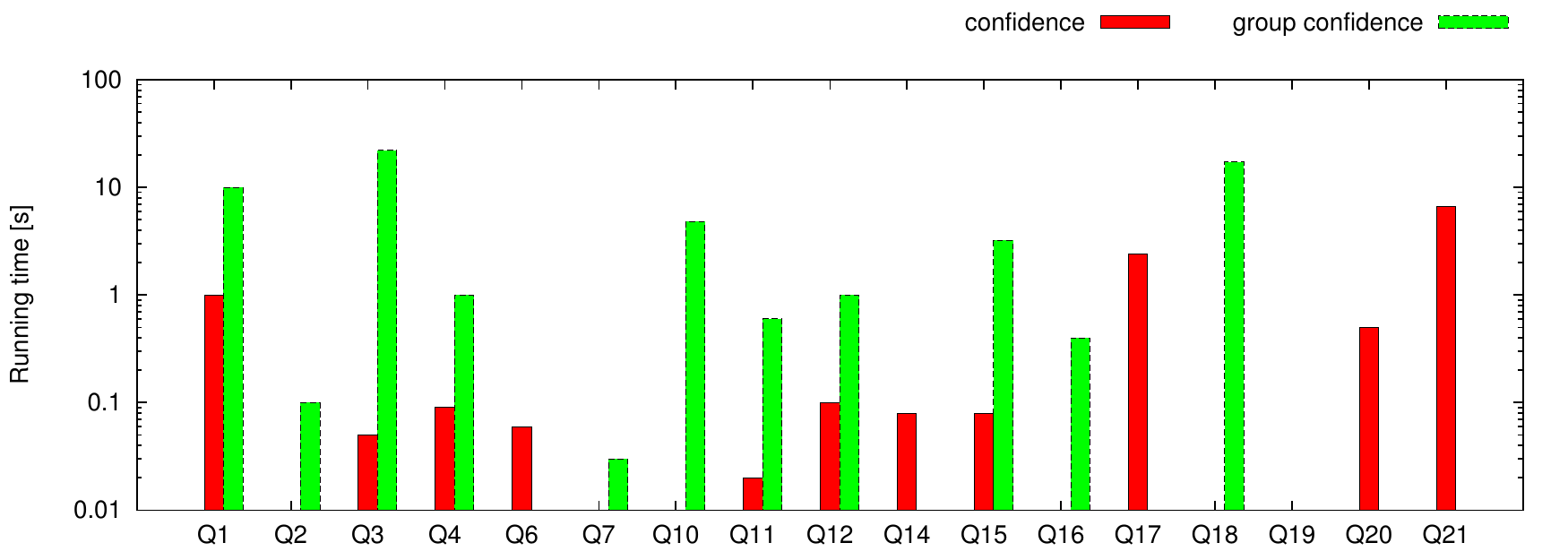}
\caption {Running time for SPROUT [26] on TPCH queries, 1GB database size} 
\label{fig:ICDE09comp}
\end{figure*}

\begin{figure}[t]
\centering
\includegraphics[scale=.7]{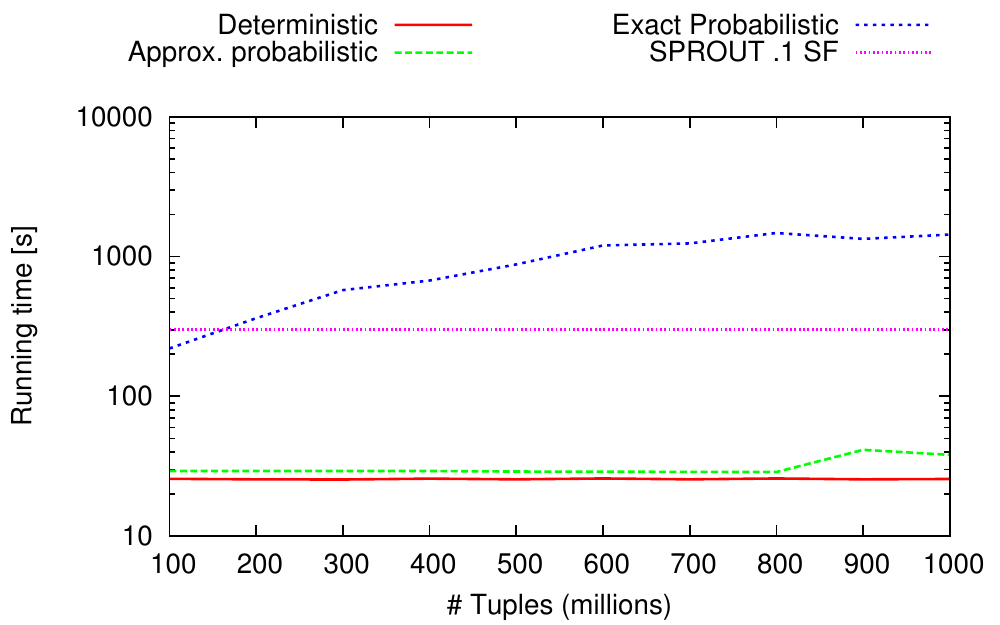}
\caption {Running time of the approximate and exact method for the COUNT aggregate.} 
\label{fig:count}
\end{figure}

\begin{figure}[t]
\centering
\includegraphics[scale=0.7]{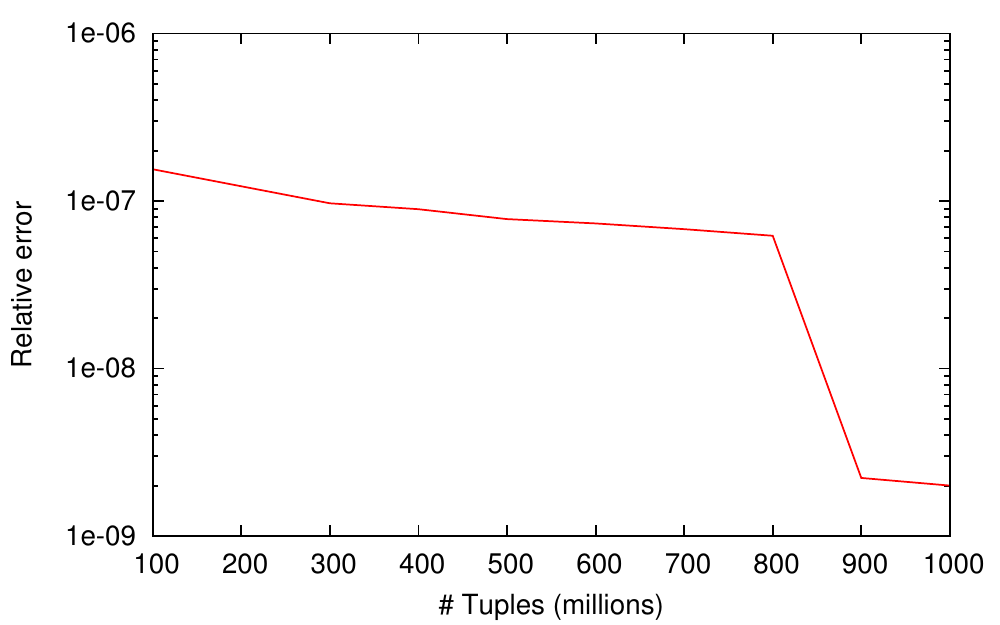}
\caption {Relative error of the approximate method for the COUNT aggregate.} 
\label{fig:error}
\end{figure}

\section{Discussion}

The theoretical and practial approach in \cite{fink} is to express the probabilistic computation as a semi-module expression/d-tree. This ensures very elegant theoretical treatment. Our theory looks very similar but there are important differences. First, by encoding the distributions as PGFs/polynomials in a fixed space we maintain uniformity of the computation that is exploited by the implementation. Second, the interpretation of the polynomial leads in a natural way to possible approximations. Third, a connection with FFT is naturally established. As our experiments revealed, the subtle theoretical differences result in large differences in performance.

\section{Conclusion}
In this paper we presented a comprehensive probabilistic database management system that can answer complex queries involving aggregates in a massive databases of 1TB size, in about 100 seconds. This represents an important performance boost compared to existing systems. Our solution is based on a uniform undelying theory based on probability generating functions, combined with performance-driven implementations and approximate solutions for extremely large input data. The system also draws its power from the architecture-centered infrastructure provided by Glade. In the future we plan to expand the system with an AVERAGE aggregate and to explore its possible applications in the area of fuzzy queries.


\bibliographystyle{abbrv}
\bibliography{pdb_blind.bib}  



\end{document}